\documentclass[sigconf]{acmart}
\usepackage[utf8]{inputenc}
\usepackage{tikz}
\usepackage{amsmath}
\usepackage{amsthm}
\usepackage{graphicx}
\usepackage{subcaption}
\usepackage[noend]{algpseudocode}
\usepackage[ruled, linesnumbered]{algorithm2e}
\usepackage{balance}

\newcommand{\incremental}{\textsc{IGC}\xspace}

\newcommand{\fullydynamic}{\textsc{FDGC}\xspace}

\newcommand{\pgdnaive}{\textsc{PGDN}\xspace}

\newcommand{\wikitalk}{\textbf{\texttt{WikiTalk}}\xspace}

\newcommand{\wikitalks}{\textbf{\texttt{WT}}\xspace}

\newcommand{\wikivote}{\textbf{\texttt{WikiVote}}\xspace}

\newcommand{\wikivotes}{\textbf{\texttt{WV}}\xspace}

\newcommand{\pokec}{\textbf{\texttt{Soc-Pokec}}\xspace}

\newcommand{\pokecs}{\textbf{\texttt{SP}}\xspace}

\newcommand{\lj}{\textbf{\texttt{soc-LiveJournal}}\xspace}

\newcommand{\ljs}{\textbf{\texttt{LJ}}\xspace}

\newcommand{\pgd}{\textsc{PGD}\xspace}

\newcommand{\remove}[1]{}
\DontPrintSemicolon

\newcommand{\butterfly}

\newtheorem{lemma}{Lemma}
\newtheorem{problem}{Problem}

\title{Efficient Batch Dynamic Graphlet Counting}
\author{Hriday G}
\email{f20191212@hyderabad.bits-pilani.ac.in}
\affiliation{\institution{BITS Pilani, Hyderabad Campus}
\city{Hyderabad}
\country{India}}
\author{Pranav Saikiran Sista}
\email{f20171225h@alumni.bits-pilani.ac.in}
\affiliation{\institution{BITS Pilani, Hyderabad Campus}
\city{Hyderabad}
\country{India}}
\author{Apurba Das}
\email{apurba@hyderabad.bits-pilani.ac.in}
\affiliation{\institution{BITS Pilani, Hyderabad Campus}
\city{Hyderabad}
\country{India}}
\date{}

\nobalance

\begin{document}

\begin{abstract}
Graphlet counting is an important problem as it has numerous applications in several fields, including social network analysis, biological network analysis, transaction network analysis, etc. Most of the practical networks are dynamic. A graphlet is a  subgraph with a fixed number of vertices and can be induced or non-induced. There are several works for counting graphlets in a static network where graph topology never changes. Surprisingly, there have been no scalable and practical algorithms for maintaining all fixed-sized graphlets in a dynamic network where the graph topology changes over time. We are the first to propose an efficient algorithm for maintaining graphlets in a fully dynamic network. Our algorithm is efficient because (1) we consider only the region of changes in the graph for updating the graphlet count, and (2) we use an efficient algorithm for counting graphlets in the region of change. We show by experimental evaluation that our technique is more than 10x faster than the baseline approach.
\end{abstract}

\maketitle

\section{Introduction}

Graphlets are induced subgraphs with significant frequency counts that can be found in any network. The graphlets which are considered to be important are those of sizes 3, 4, and 5  and can be identified in both directed and undirected graphs. Given an undirected graph $G$, with $n$ vertices and $m$ edges, counting graphlets of size $k$ is defined as counting the number of subgraphs isomorphic to subgraphs of a fixed size $k$. It is a computationally intensive problem since there are $O(|V|^k)$ possible subgraphs of size $k$. Despite that, graphlet counting is an important problem due and it has served as a building block in solving problems in several areas.

In computational biology, they have been used for detecting cancer cells~\cite{milenkovic2010systems}, and for analyzing protein-protein interactions~\cite{milenkovic2008uncovering, prvzulj2004modeling}. Other areas where they have been used include bioinformatics~\cite{VI+-JCB-2010,WH+-TCBB-2012,SH-CN-2015,WZ+-TKDE-2017}, computer vision~\cite{HB-CVPR-2007,ZS+-TIP-2012,ZS+-CVPR-2013}, social network analysis~\cite{SPT-ICDM-2013,RBA-TKDE-2014,WL+-TKDD-2014, janssen2012model} and so on. They have also been used to compare graphs using graph kernels~\cite{shervashidze2009efficient}.

Although most research involving graphlets has been done on static graphs, progress has been made on dynamic networks. This is pretty significant since this will open up ways to analyze networks that change over time. Solving such problems will give graphlets more real-world applications since most problems are dynamic in nature. This is especially true in the case of biological networks. Knowing the counts of graphlets in a biological network can provide information about the network's properties. However, these networks can frequently change over time. For example, when a stem cell differentiate into other cell types during its development process~\cite{mukherjee2018counting}, or when chromosomes' chromatin structures change due to various events~\cite{mukherjee2018counting}. Due to the ever-changing nature of these networks, maintaining the counts of graphlets statically will not be sufficient as they will be outdated. It is necessary to update the counts whenever the network changes. Similar problems are present in other areas as well. For instance, social networks are constantly changing with new connections being added and existing connections are becoming obsolete, or new transactions are getting added and old transactions are getting removed from the historical data store in transaction databases.

In order to solve such problems, we have proposed an algorithm that maintains the counts of 4-node graphlets in a network which is changing over time due to the addition and deletion of edges.
\section{Related Works}
Most of the previous works have focused on counting only a special type of graphlets, such as cliques and cycles~\cite{k-clique-listing, triangle-counting-stream-1, triangle-four-cycle-counting, fully-dynamic-triangle-count}. However, counting all graphlets of a fixed size has been used extensively in recent times in graph mining and machine learning applications. We discuss the state-of-the-art algorithms for graphlet counting in static and dynamic graphs.

\noindent\textbf{Graphlet counting in static network:}~Counting graphlets exactly~\cite{graphlet-exact-1, PSV-WWW-2017} and approximately~\cite{graphlet-estimate-1,RBA-TKDE-2014,WZ+-TKDE-2017, path-sampling} in a static network has been studied extensively. In~\cite{graphlet-exact-1}, the authors develop an algorithm for counting graphlets up to size $4$. In their work, they iterate over edges to compute smaller-size graphlets and use combinatorial arguments to extend them to large graphlets in constant time. Their algorithm is scalable and easy to parallelize. In~\cite{PSV-WWW-2017}, the authors have developed an algorithm for exactly counting all graphlets of size $5$. Their framework is based upon cutting a big subgraph pattern into smaller ones and using the smaller patterns to get the counts of larger patterns. These two algorithms are state-of-the-art for exact counting up to size $5$. Approximate graphlet counting algorithms are developed based on techniques such as path sampling~\cite{path-sampling}, edge and local neighborhood sampling~\cite{graphlet-estimate-1, RBA-TKDE-2014}, subgraph sampling~\cite{WZ+-TKDE-2017}.

\noindent\textbf{Graphlet counting in dynamic network:}~Over the last few years, a specific type of subgraph such as cycle~\cite{dynamic-cycle-1,fully-dynamic-triangle-count, dynamic-cycle-2}, clique~\cite{batch-dynamic-clique,fully-dynamic-clique} has been studied extensively. Hanauer, Henzinger, \& Hua~\cite{fully-dynamic-graphlet} first proposed an algorithm for maintaining the number of all graphlets of size $4$ and provided rigorous theoretical analysis on the bounds on update and query complexity, update time, query time, and space complexity. Our work differs from theirs in several aspects: (1) Our algorithm is practical - it is based on one of the most efficient algorithms for exact counting $4$ node graphlets in practice such as \textit{PGD}. (2) Space complexity of our algorithm is typically low as opposed to the work of Hanauer et al. where the space complexity is $O(n^2)$ as it requires maintaining several data structures along with the dynamic graph in their work. In contrast, we only need to maintain the dynamic graph. (3) Our algorithm works in a batch dynamic mode, meaning we process all the edges in a batch at a time instead of updating one edge at a time.

\remove{
Two algorithms which were valuable in conducting our research were \cite{ANRD-ICDM-2015}, which contained efficient methods for maintaining the count of graphlets in a static graph, and \cite{CL-ASONAM-2017}, which contains an efficient algorithm to maintain graphlets when edges are added one at a time. The algorithm in \cite{ANRD-ICDM-2015} is part of our algorithm for getting the graphlet counts. The algorithm in \cite{CL-ASONAM-2017} has been used for an alternate algorithm, where \cite{CL-ASONAM-2017} is run on every edge in a batch. This algorithm is compared with the main algorithm we are proposing.
}
\section{Problem Definition}

We consider simple undirected graphs and model the dynamic graph as a sequence of edge additions and deletions.

We consider monitoring a graph that changes over time. Assume that, for any $t\geq 0$, $G^t = (V^t,E^t)$ be the graph up to time $t$ where $V^t$ is the set of vertices, and $E^t$ is the set of edges. For any $t\geq 0$, at time $t+1$, we add a tuple $<e,o>$ from the stream, where $o\in\{+,-\}$ represents an update process, add edge or delete edge, $e = (u,v)$ a pair of vertices. A graph $G^{t+1} = (V^{t+1}, E^{t+1})$ is generated by adding a new edge or deleting an existing edge as follows:
\[
        E^{t+1}= 
        \begin{cases}
            E^{t}\cup\{u\} & \textit{if } o=+ \\
            E^{t}\setminus\{u\} & \textit{if } o=-
        \end{cases}
\]

The addition and deletion of vertices can be handled similarly. Moreover, we assume that by adding an edge $e = (u,v)$, the vertices $u$ and $v$ are added to the graph if they are not present at time $t$. Furthermore, when a vertex is deleted, we assume all the incident edges are deleted before deleting the vertex. For simplicity, the rest of the paper deals with edge addition and deletion - vertex operations can be handled easily by iteratively deleting all the edges adjacent to a vertex.

The neighborhood of a vertex $u\in V^t$ at time $t$ is defined as $N^t(u) = \{v | (u,v)\in E^t\}$. Similarly, the $k$-hop neighborhood of $u$ at time $t$ is denoted as $N^t_k(u)$. When the context is clear, we will use $N(u)$ and $N_k(u)$ to mean neighborhood and $k$-hop neighborhood.

Given a simple undirected graph $G=(V,E)$, a graphlet of size $k$ nodes is defined as a subgraph $g_k\subset G$ where $g_k$ has $k$ vertices. A graphlet $g_k$ is an \textbf{induced graphlet} if $g_k$ is an induced subgraph of $G$ and it is a \textbf{connected graphlet} if $g_k$ is a connected subgraph. Given a subset $S\subset V$, $G[S]$ denotes the subgraph of $G$ induced by 
$S$.

Next, we define the problem of maintaining the frequency of connected induced graphlets of size-$k$. 

\begin{problem}
    Given an evolving graph $G^t=(V,E^t)$, an integer $k$, and a batch of edges $\mathcal{B} = \mathcal{A}\cup\mathcal{D}$, maintain the frequency of connected induced size-$k$ graphlets in $G^{t+1} = (V,E^t\cup\{\mathcal{A}\}\setminus\{\mathcal{D}\})$. 
\end{problem}

We focus on a fully dynamic setting when the addition and deletion of edges are allowed in an arbitrary order in a batch. We aim to maintain the frequency of all connected induced size-$4$ graphlets (shown in Figure~\ref{fig:graphlets}) without recomputing from scratch.

\remove{
Given an undirected graph $G = (V, E)$ with $n$ vertices and $m$ edges, graphlets are connected subgraph patterns. {Fig.~\ref{fig:graphlets} shows all 4-node graphlets}. The dynamic $k$-graphlet counting problem maintains the number of graphlets in the graph upon edge insertions and deletions, given individually or in a batch. {\color{red} In this work, we consider the incremental settings where only additions are supported}.
}

\begin{figure}[h]
    \centering
    \includegraphics[scale=0.4]{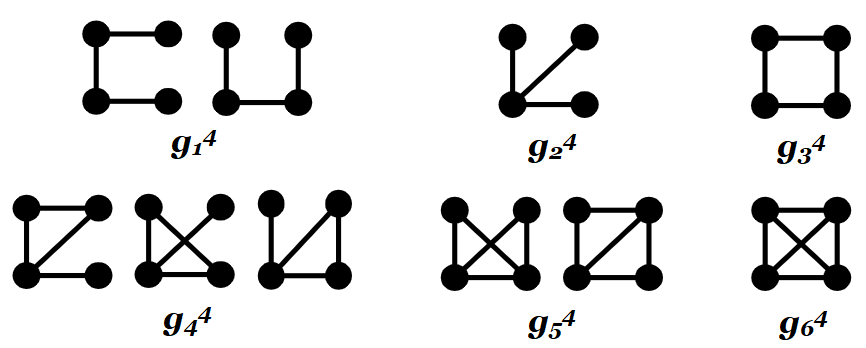}
    \caption{A list of different kinds of graphlets from the perspective of the bottom edge. From top to bottom, left to right, the graphlets are as follows: 3-paths ($g_1^4$), a 3-star ($g_2^4$), a 4-cycle ($g_3^4$), tailed triangles ($g_4^4$), diamonds ($g_5^4$), and a 4-clique ($g_6^4$).}
    \label{fig:graphlets}
\end{figure}

In this work, we develop algorithms for the maintenance of size-$4$ graphlets. In doing so, we identify the local subgraph and count graphlets in the local subgraph using a state-of-the-art deterministic graphlet counting algorithm such as {\tt PGD}.

\noindent\textbf{PGD Algorithm:}~ The algorithm counts graphlets of size $3$ and $4$ in a static simple undirected graph. The algorithm leverages combinatorial arguments for different graphlets.

The algorithm iterates over the edges of the graph. Each iteration counts only a few graphlets and uses combinatorial arguments to derive the exact count of the rest of the graphlets in constant time. We leverage this algorithm to maintain the graphlet count in a dynamic network by appropriately choosing the subgraph and applying PGD on that subgraph with a guarantee that we can maintain the global graphlet count by counting graphlets in that subgraph only.

\remove{
In this paper, we analyze our algorithms in the work-depth model, where the \textbf{work} of an algorithm is defined to be the total number of operations performed, and the \textbf{depth} is defined to be the longest sequential dependency. We assume that the concurrent read and write are supported with $O(1)$ work and depth. A \textbf{work-efficient} parallel algorithm is one with work that asymptotically matches the best-known sequential time complexity for the problem.
}
\section{Algorithms}

In this section, we design efficient algorithms for maintaining size-$4$ graphlets in a dynamic graph. First, we will describe the algorithm for an incremental batch where the goal is to update the graphlet count once a batch of new edges is added to the graph. Next, we will describe the algorithm for a fully dynamic batch consisting of edges for the addition and edges for the deletion. The algorithm's objective is to maintain the frequency of all induced graphlets in the dynamic graph.

\remove{
Consider the addition of a new batch of edges $\mathcal{B}$ to the graph $G=(V,E)$. Assume that $\mathcal{F}_4$ is the frequency vector of all the size-$4$ induced graphlets. When $G$ is updated to $G' = G+\mathcal{B}$, a set of new graphlets $\mathcal{N}_4$ will be generated, and a set of existing graphlets $\mathcal{D}_4$ will be subgraphs of the new graphlets that we call subsumed graphlets. Suppose the frequency vector for the new graphlets is denoted by $\mathcal{F}^{N}_4$, and the frequency vector for the subsumed graphlets is denoted by $\mathcal{F}^{D}_4$. The updated graphlet frequency vector $\mathcal{F}'_4 = \mathcal{F}_4 + \mathcal{F}^{N}_4 - \mathcal{F}^{D}_4$. We have used several algorithms to demonstrate how existing methods can solve this problem and how the algorithm we are proposing will be better.
}

\subsection{Baseline Algorithm}

A straightforward algorithm for maintaining graphlet counts is to count all the graphlets in the updated graph using a state-of-the-art graphlet counting algorithm such as PGD~\cite{ANRD-ICDM-2015}. We name this baseline as \pgdnaive. This counting-from-scratch strategy is inefficient when the batch of edges touches a tiny portion of the graph. At the same time, the static algorithm for graphlet counting has to work on the much larger region of the graph for graphlet counting, which is required for maintaining the global count. An efficient technique would be to update the graphlet based on the count from the graph region where the change is located.

\remove{
The state-of-the-art algorithm for maintaining the graphlets in a dynamic graph deals with single-edge addition~\cite{CL-ASONAM-2017}. We tried to modify this algorithm to work on the batch update instead of a single-edge update. Still, we found it difficult as it is unclear how to keep track of the changes in the graphlet due to adding more than one edge when we consider the maintenance of induced graphlet counts.
}

\remove{
\subsection{Single Edge Graphlet Counting Algorithm}
Another approach that was considered was processing each edge in a newly added batch one by one and updating the graphlet count accordingly. This was done by using \cite{CL-ASONAM-2017}, which counted the number of graphlets modified by the edge of one edge to a graph. We devised an algorithm which used \cite{CL-ASONAM-2017} over each new edge in a batch to update the global graphlet count. The results of this algorithm were compared against the main algorithm which we are proposing in this paper.

\begin{algorithm}

Initialization:\;
Global Graphlet Counts $\mathcal{F}_4$ = [0, 0, 0, 0, 0, 0]\;

\ForEach{Batch $\mathcal{B} \in Inputs$}{

\ForEach{Edge $e \in B$}{
$\mathcal{G}\gets\mathcal{G}+e$\;
$\mathcal{C}' \gets \mathcal{C}_a + ASONAM(\mathcal{G}, e)$\;
$\mathcal{F}_4 \gets \mathcal{F}_4 + \mathcal{C}'$\;
}
}
    \Return $\mathcal{F}_4$\; 
\caption{Single Edge Graphlet Counting Algorithm}
\label{algo-ASONAM}
\end{algorithm}
}

\subsection{Incremental Stream}
Now we will present an efficient algorithm for updating the graphlet count once a batch of edges $\mathcal{B}$ is added to the original graph $G$ through systematically exploring the subgraphs local to the changes in the graph.  Once the subgraphs are explored, we count the graphlets in the local subgraphs. Finally, we aggregate these local counts to maintain the global graphlet counts in the updated graph. 

Suppose a batch of new edges $\mathcal{B} = \{e_1, e_2, ..., e_l\}$ is added to the graph $G = (V, E)$ and the updated graph is $G' = G+\mathcal{B} = (V,E\cup\mathcal{B})$. This algorithm maintains two sub-graphs $\mathcal{G}_a$ and $\mathcal{G}_b$ of $\mathcal{G}$. It is easy to see that the diameter of size-$4$ graphlet is $3$, which occurs in the case of a $3$-path (Figure~\ref{fig:graphlets}). Hence, the construction of $\mathcal{G}_a$ is as follows: For each edge $e_i = (u_i, v_i)$, we construct a subgraph $g_i$ induced by the vertex set $W_i = N_3(u_i)\cup N_3(v_i)$. Finally, $\mathcal{G}_a = g_1\cup g_2\cup...\cup g_l$. We construct $\mathcal{G}_b$ by removing all the edges in $\mathcal{B}$ from $\mathcal{G}_a$, formally, $\mathcal{G}_b = \mathcal{G}_a\setminus\mathcal{B}$.

The pseudocode is presented in Algorithm~\ref{algo-1}. The algorithm first updates the original graph $\mathcal{G}$ by adding the batch $\mathcal{B}$ to it. Next, it creates $\mathcal{G}_a$ and $\mathcal{G}_b$ and counts graphlets in $\mathcal{G}_a$ and $\mathcal{G}_b$. The following lemma shows that we can maintain the exact count of the graphlets using the graphlet counts in $\mathcal{G}_a$ and $\mathcal{G}_b$.\\

\begin{algorithm}[ht!]

Initialization:\;
Global Graphlet Counts $\mathcal{F}_4$ = [0, 0, 0, 0, 0, 0]\;
$\mathcal{G}_a$ Graphlet Counts $\mathcal{C}_a$ = [0, 0, 0, 0, 0, 0]\;
$\mathcal{G}_b$ Graphlet Counts $\mathcal{C}_b$ = [0, 0, 0, 0, 0, 0]\;

\ForEach{Batch $\mathcal{B} \in Inputs$}{

$\mathcal{G}\gets\mathcal{G}+\mathcal{B}$\;
Generate graph $\mathcal{G}_a$ = ($\mathcal{V}_a$, $\mathcal{E}_a$)\;
Generate graph $\mathcal{G}_b$ = ($\mathcal{V}_b$, $\mathcal{E}_b$)\;

$\mathcal{C}_a \gets PGD(\mathcal{G}_a)$\;
$\mathcal{C}_b \gets PGD(\mathcal{G}_b)$\;
$\mathcal{C}' \gets \mathcal{C}_a - \mathcal{C}_b$\;
 $\mathcal{F}_4 \gets \mathcal{F}_4 + \mathcal{C}'$\;
}
    \Return $\mathcal{F}_4$\; 
\caption{\incremental: Incremental Graphlet Counting}
\label{algo-1}
\end{algorithm}

\begin{lemma}\label{lemma-1}
Suppose a batch of new edges $\mathcal{B}$ is added to the graph $\mathcal{G} = (\mathcal{V}, \mathcal{E})$ and the graph is updated to $\mathcal{G}'$. Then, $\mathcal{F}^{\mathcal{G}'}_4 = \mathcal{F}^{\mathcal{G}}_4 + \mathcal{C}_a-\mathcal{C}_b$.
\end{lemma}

\begin{proof}
There are three cases to consider: 

\textbf{case-1:}~all graphlets containing at least one new edge are added to the count. Observe that $\mathcal{C}_a$ captures all the graphlets that contain at least one new edge. Further, $\mathcal{C}_b$ will never capture any of these graphlets since they are missing some edges in $\mathcal{G}_b$.  So, the count of every graphlet containing at least a new edge is added exactly once to the global count to update the answer.

\textbf{case-2:}~all graphlets extended to a new graphlet by addition of edges are removed from the count. Observe that $\mathcal{C}_b$ captures all the graphlets as we remove all the new edges from $\mathcal{G}_a$ to get $\mathcal{G}_b$ and $\mathcal{C}_a$ will never capture the count of any such graphlets as they become non-induced once a batch is added. So, the count of these graphlets is subtracted exactly once.

\textbf{case-3:}~ All graphlets that do not change after the addition of new edges do not contribute to the changes in the global graphlet counts. Since the graphlets were present before and after the addition of edges, they must be captured by both $\mathcal{C}_a$ and $\mathcal{C}_b$. So the count of every such graphlet is added exactly once and subtracted exactly once. So, the global count does not change.
\end{proof}

\remove{
\subsection{Efficient Algorithm - Decremental Maintenance}

Suppose a batch $D$ of existing edges is removed from the graph $G=(V, E)$. We can update the graphlet counts through a reduction to the additional case. Assume that the edges in $D$ are being added instead of getting removed. Then following Lemma-1, $\mathcal{C}_a-\mathcal{C}_b$ is the number of graphlets that contain at least one edge in $D$. Thus, instead of adding, we can subtract $\mathcal{C}_a-\mathcal{C}_b$ from the current graphlet counts at Line-$12$ of Algorithm~\ref{algo-1} to get an updated count when a batch $D$ is deleted.
}

\subsection{Fully Dynamic Stream}

In this section, we describe our algorithm for maintaining global graphlet count in a fully dynamic setting where both insertion and deletion of edges are possible.



Suppose $\mathcal{B} = \mathcal{A}\cup\mathcal{D}$ is a batch of edges where $\mathcal{A}$ is a batch of edges to add and $\mathcal{D}$ is a batch of edges to delete. and Let, $H = \cup_{e = (u,v)\in\mathcal{A}\cup\mathcal{D}}G[\mathcal{N}_3(u)\cup\mathcal{N}_3(v)]$. Next we define $\mathcal{G}_a = (H\setminus\mathcal{D})\cup\mathcal{A}$ and $\mathcal{G}_b = (H\setminus\mathcal{A})\cup\mathcal{D}$ and then count the number of graphlets in $\mathcal{G}_a$ and $\mathcal{G}_b$. We use these counts in a systematic way to update the global graphlet count after processing the batch $\mathcal{B}$. The pseudocode is presented in Algorithm~\ref{fully-dynamic}

\begin{algorithm}[ht!]

Initialization:\;
Global Graphlet Counts $\mathcal{F}_4$ = [0, 0, 0, 0, 0, 0]\;
$\mathcal{G}_a$ Graphlet Counts $\mathcal{C}_a$ = [0, 0, 0, 0, 0, 0]\;
$\mathcal{G}_b$ Graphlet Counts $\mathcal{C}_b$ = [0, 0, 0, 0, 0, 0]\;

\ForEach{Batch $\mathcal{B} = \mathcal{A}\cup\mathcal{D} \in Inputs$}{

$H = \cup_{e\in\mathcal{A}\cup\mathcal{D}}G[\mathcal{N}_3(e)]$\;
$\mathcal{G}_a = (H\setminus\mathcal{D})\cup\mathcal{A}$\;
$\mathcal{G}_b = (H\setminus\mathcal{A})\cup\mathcal{D}$\;

$\mathcal{C}_a \gets PGD(\mathcal{G}_a)$\;
$\mathcal{C}_b \gets PGD(\mathcal{G}_b)$\;
$\mathcal{C}' \gets \mathcal{C}_a - \mathcal{C}_b$\;
 $\mathcal{F}_4 \gets \mathcal{F}_4 + \mathcal{C}'$\;
 $\mathcal{G}\gets\mathcal{G}\setminus\mathcal{D}\cup\{\mathcal{A}\}$\;
}
    \Return $\mathcal{F}_4$\; 
\caption{\fullydynamic: Fully Dynamic Graphlet Counting}
\label{fully-dynamic}
\end{algorithm}

\begin{lemma}
    Algorithm~\ref{fully-dynamic} correctly updates the global graphlet frequency.
\end{lemma}

\begin{proof}

\remove{
    As an alternate approach to Algorithm~\ref{fully-dynamic}, assume that we first consider deletions of edges from $\mathcal{D}$ and then the addition of edges from $\mathcal{A}$. Algorithm~\ref{fully-dynamic} should produce the same result using this alternate approach. Due to the deletion of edges in $\mathcal{D}$, some graphlets will be removed from the global count. Also, some graphlets (which are non-induced during the presence of $\mathcal{D}$ but induced due to the absence of $\mathcal{D}$) will be transitioned in to be added to the global count. Consider an edge $e^{-}\in\mathcal{D}$. Clearly, $\mathcal{C}_b$ counts all induced graphlets containing $e^{-}$ and $\mathcal{C}_a$ counts all induced graphlets without $e^{-}$ which are non-induced in the presence of $e^{-}$. Thus, $\mathcal{C}_a - \mathcal{C}_b$ adds the counts of all induced graphlets, which become non-induced in the presence of $e^{-}$. Next, assume that $e^{+}\in\mathcal{A}$ has been added to the graph. Due to this edge addition, some induced graphlets become non-induced in the presence of $e^{+}$, and they will be transitioned out, and corresponding counts will be subtracted from the global graphlet counts, which are captured in $\mathcal{C}_b$ and the newly induced graphlet count containing $e^{+}$ will be added to the global graphlet counts, which is captured by $\mathcal{C}_a$. Hence the results follow, and the proof is complete.
}

    There are four types of graphlets. (1) {\tt type-1} graphlets containing edges in $\mathcal{D}$ but no edges in $\mathcal{A}$; (2) {\tt type-2} graphlets containing  edges in both $\mathcal{A}$ and $\mathcal{D}$ (3) {\tt type-3} graphlets containing edges in $\mathcal{A}$ but no edge in $\mathcal{D}$ (4) {\tt type-4} graphlets containing edges neither from $\mathcal{A}$ nor from $\mathcal{D}$. All {\tt type-4} graphlets are captured by both $\mathcal{C}_a$ and $\mathcal{C}_b$, So $\mathcal{C}_a-\mathcal{C}_b = 0$. Counts due to {\tt type-3} graphlets are captured in $\mathcal{C}_a$, added to the global count, but they are not captured by $\mathcal{C}_b$. {\tt type-2} graphlets are captured once in $\mathcal{C}_a$ and once in $\mathcal{C}_b$. So $\mathcal{C}_a-\mathcal{C}_b = 0$.  {\tt type-1} graphlets are removed, and therefore, corresponding counts are deleted, captured in $\mathcal{C}_b$, but not captured in $\mathcal{C}_a$. Transitioning of the graphlets as in \textbf{case-2} of Lemma~\ref{lemma-1} are handled similarly as in Lemma~\ref{lemma-1}. This completes the proof.
    
\end{proof}

\section{Experiments}

We evaluated our algorithm on large datasets to show the efficiency over the baseline \pgdnaive. Every time we add a batch of edges and perform recomputation from scratch using the \pgd. We show the description of the dataset in Table~\ref{table:dataset} taken from the Stanford SNAP~\cite{snapnets}. We implemented all the algorithms in C++ and evaluated them in a system equipped with Intel(R) Xeon(R) W-2223 CPU $@$ $3.60$GHz with $4$ physical cores and $32$G RAM. We execute \pgd with OpenMP to run in parallel.

\begin{table}[ht!]
\small
\caption{\label{table:dataset}
Characteristics of the Data sets}
\vspace*{-4mm}
\begin{center}
\scalebox{0.888}{
\begin{tabular}{|c|c|c|}
\hline
\textbf{Dataset} & $|V|$  & $|E|$ \\
\hline
\wikitalk (\wikitalks) & $2394385$ & $4659565$ \\
\hline
\wikivote (\wikivotes) & $7115$ & $100762$ \\
\hline
\pokec (\pokecs) & $1632803$ & $22301964$ \\
\hline
\lj (\ljs) & $4033137$ & $27933062$ \\
\hline
\end{tabular}}
\end{center}
\end{table}

\noindent\textbf{Batch creation:}~For the incremental batch computation, we group edges in a batch and execute the algorithms by inserting one batch at a time. For the fully dynamic batch, we create the batches as follows: First, we give $+$ive and $-$ive labels to the edges by sampling each edge with probability $p$ for selecting it to be a positive edge and putting the sampled edges in a queue. Later, we extract edges from the queue and sample them with probability $1-p$ to assign a $-$ive sign. This way, we create both positive and negative edges where the positive edge refers to the addition and the negative refers to the deletion. Next, we group them as before for creating batches. For our experiment, we choose $p=0.7$.

\noindent\textbf{Computation on incremental batch:}~We empirically evaluated our algorithm \incremental to show that \incremental is much faster than \pgdnaive on large graphs such as \pokec and \lj. This is because \incremental counts graphlets on a relatively small size subgraph, which is local to the changes in the graph due to the batch update, while \pgdnaive access the entire network built so far to count the graphlets. Also, the performance of \incremental and \pgdnaive is almost similar on \wikivote (Figure~\ref{fig:wiki-vote-incr-10}) and \wikitalk (Figure~\ref{fig:wiki-talk-incr-10}) because these networks are relatively small and sparse. Also, \incremental needs to compute $3$-hop neighborhood for constructing the local subgraph, which is not substantially smaller than the original network where \pgd work in the baseline. The results on batches of size $10$ and $100$ are shown in Figure~\ref{fig:incr-batch-10} and Figure~\ref{fig:incr-batch-100}, respectively. In both cases, we observe that as the size of the graph increases with the increase in the edge counts, \incremental is substantially faster than the baseline \pgdnaive in two large networks \pokec (Figure~\ref{fig:soc-pokec-incr-10} and Figure~\ref{fig:pokec-100}) and \lj (Figure~\ref{fig:live-journal-incr-10} and Figure~\ref{fig:live-journal-100}). This is as expected because \pgdnaive accesses a larger graph each time compared to the size of the graph accessed by \incremental. 

\noindent\textbf{Computation on fully dynamic batch:}~In a fully dynamic batch, we first removed the edges, which are both added and deleted, then executed Algorithm \fullydynamic on the rest of the edges. Similar to the incremental stream, we observed similar behavior on two large networks \pokec and \lj as shown in Figure~\ref{fig:fd-batch-100}. Similar to the earlier observation, as the network size increases, the runtime of \fullydynamic is significantly faster compared to \pgdnaive than the initial state of the graph. 

\begin{figure}[t!]
	\centering
	\vspace{-3ex}
        \begin{minipage}{.5\textwidth}
	\subfloat[][{\wikivote}]{\includegraphics[scale=0.26]{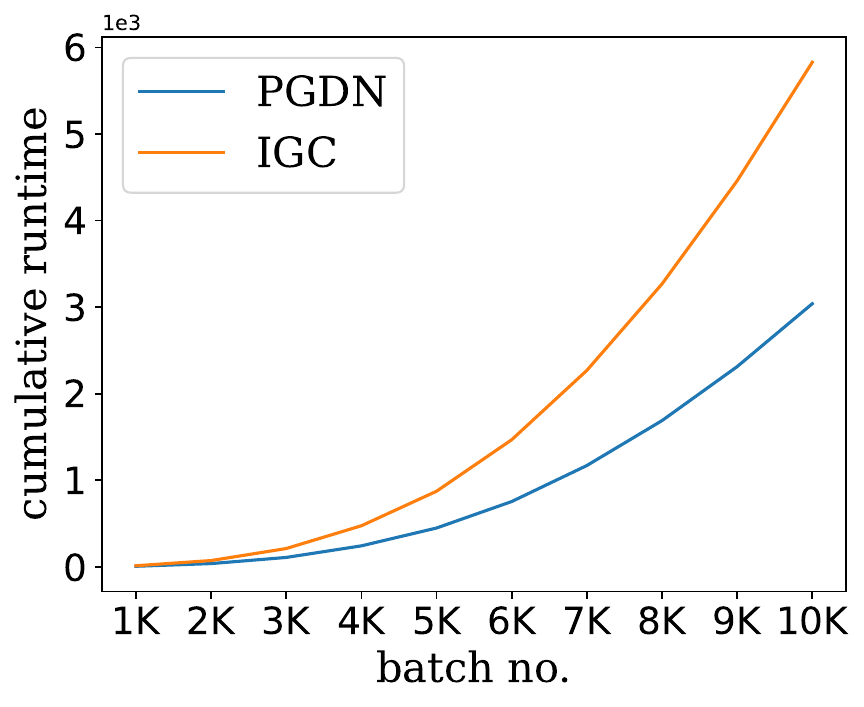}\label{fig:wiki-vote-incr-10}}~{}~{}
	\hspace{1ex}
	\subfloat[][{\wikitalk}]{\includegraphics[scale=0.26]{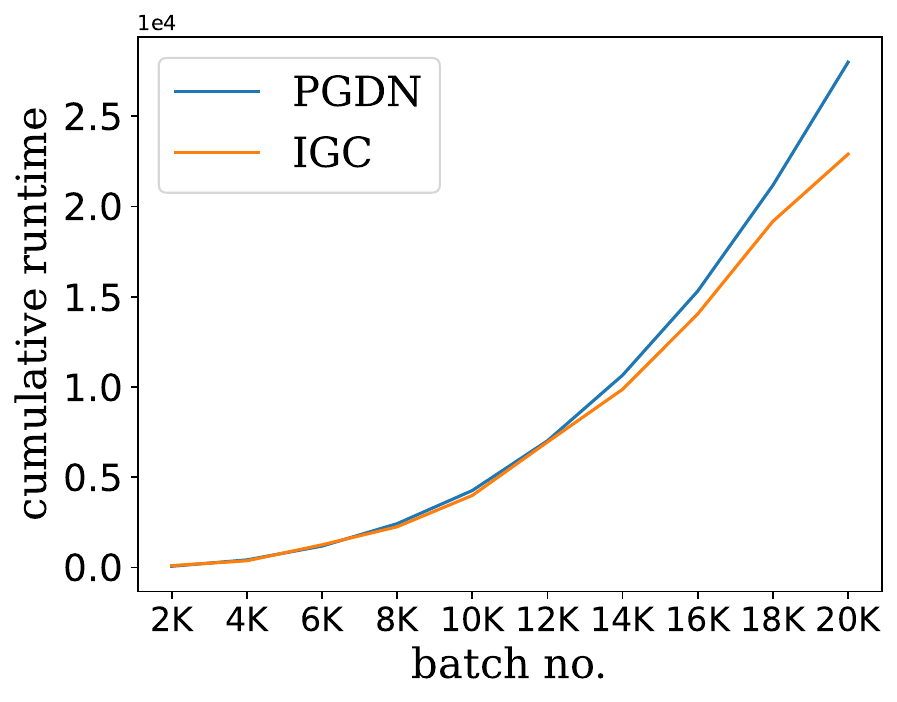}\label{fig:wiki-talk-incr-10}}~{}~{}
        \end{minipage}
        \begin{minipage}{.5\textwidth}
	\subfloat[][\pokec]{\includegraphics[scale=0.26]{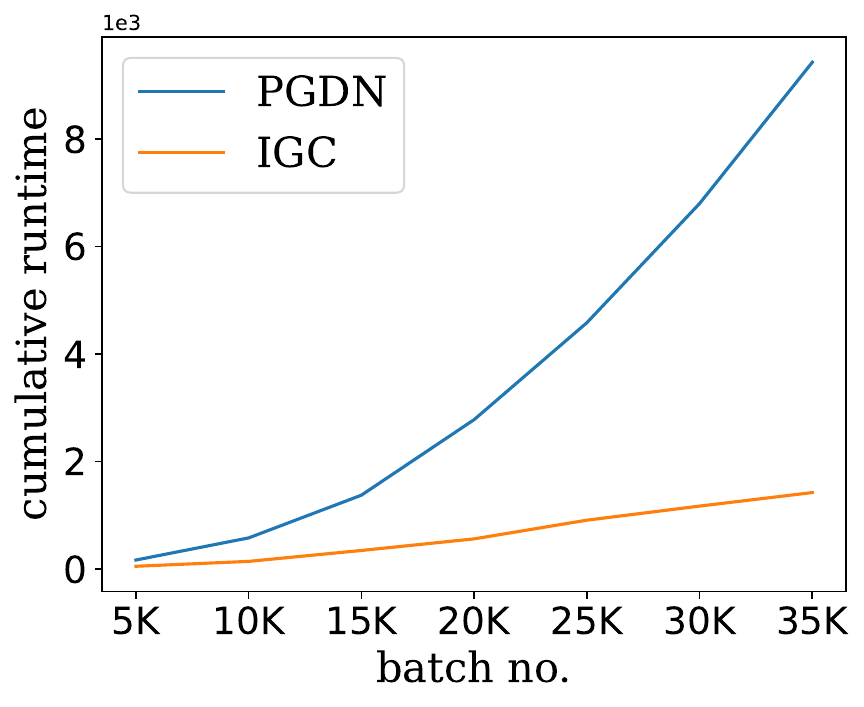}\label{fig:soc-pokec-incr-10}}~{}~{}
	\hspace{1ex}
	\subfloat[][\lj]{\includegraphics[scale=0.26]{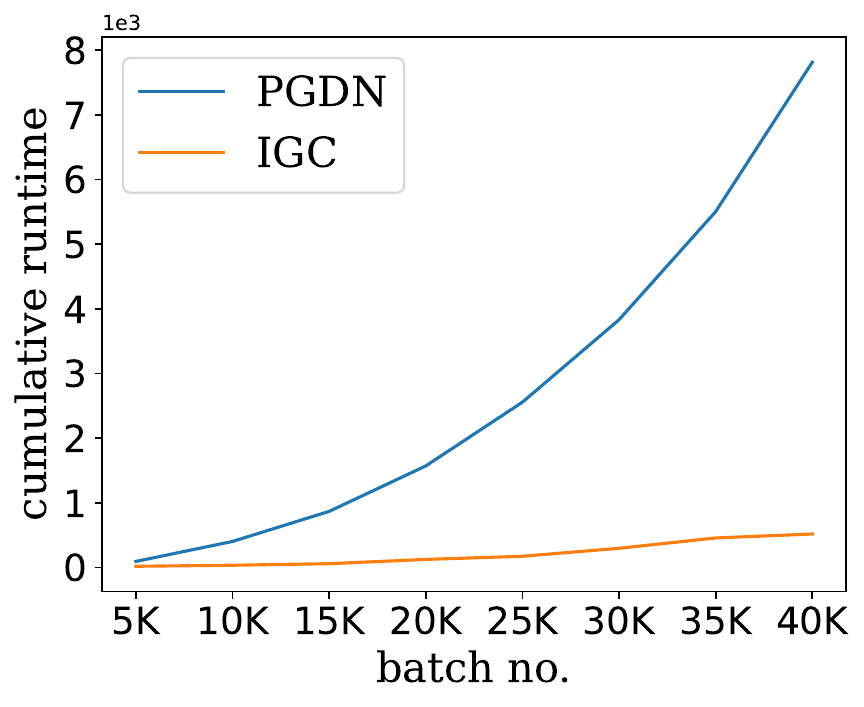}\label{fig:live-journal-incr-10}}~{}~{}
        \end{minipage}
	\vspace{-1ex}
	\caption{ Cumulative time to add insert only batches (of size $10$). Time at $y$-axis shows the overall time taken to add all the batches so far (represented by the batch number on the $x$-axis)
	}\label{fig:incr-batch-10}
	\vspace{-3ex}
\end{figure}

\begin{figure}[h]
\centering
 \begin{minipage}{.5\textwidth}
  \subfloat[][{\lj}]{\includegraphics[scale=0.3]{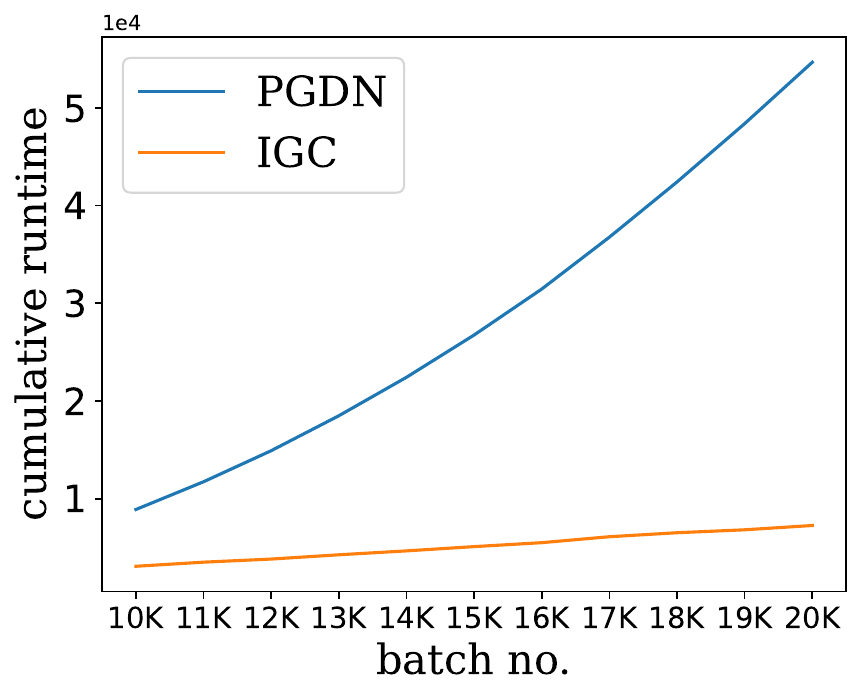}\label{fig:live-journal-100}}~{}~{}
  \subfloat[][{\pokec}]{\includegraphics[scale=0.3]{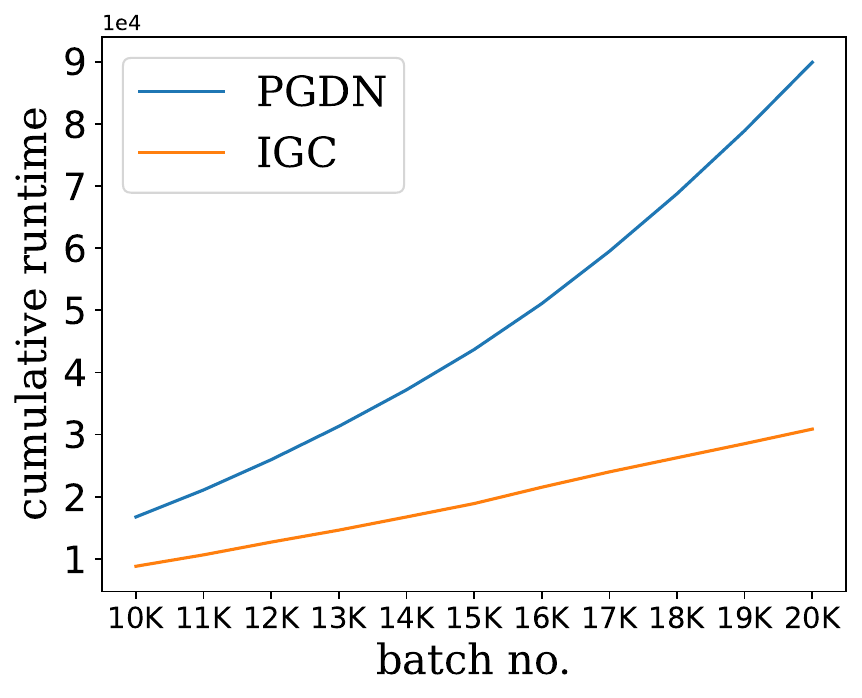}\label{fig:pokec-100}}~{}~{}
  \caption{Cumulative time to add insert only batches (of size $100$). Time at $y$-axis shows the overall time taken to add all the batches so far (represented by the batch number on the $x$-axis)}\label{fig:incr-batch-100}
 \end{minipage}
\end{figure}

\begin{figure}[h]
\centering
 \begin{minipage}{.5\textwidth}
  \subfloat[][{\lj}]{\includegraphics[scale=0.3]{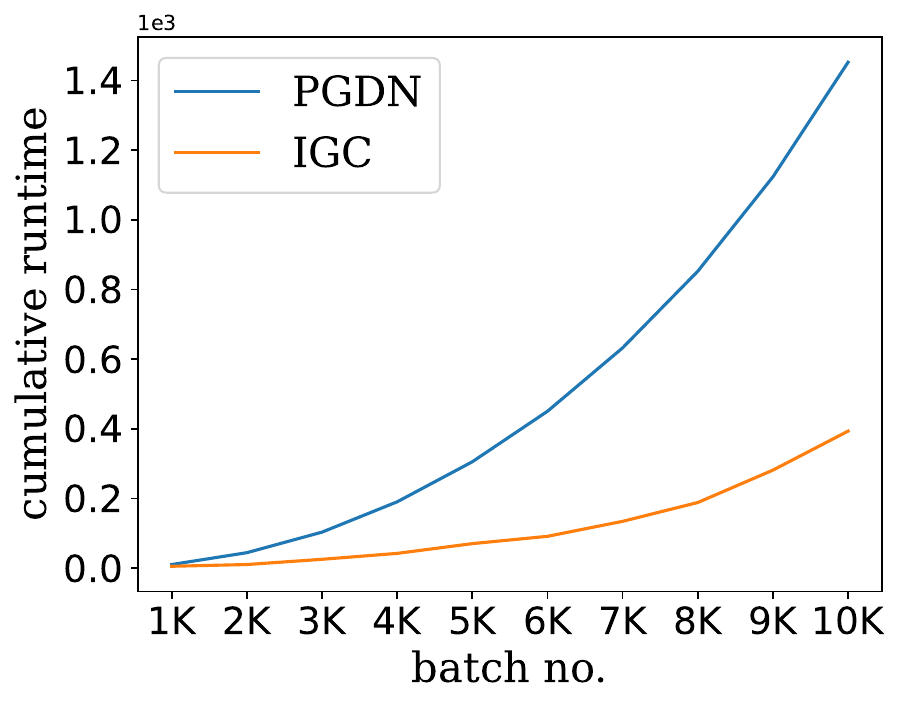}\label{fig:live-journal-100-fd}}~{}~{}
  \subfloat[][{\pokec}]{\includegraphics[scale=0.3]{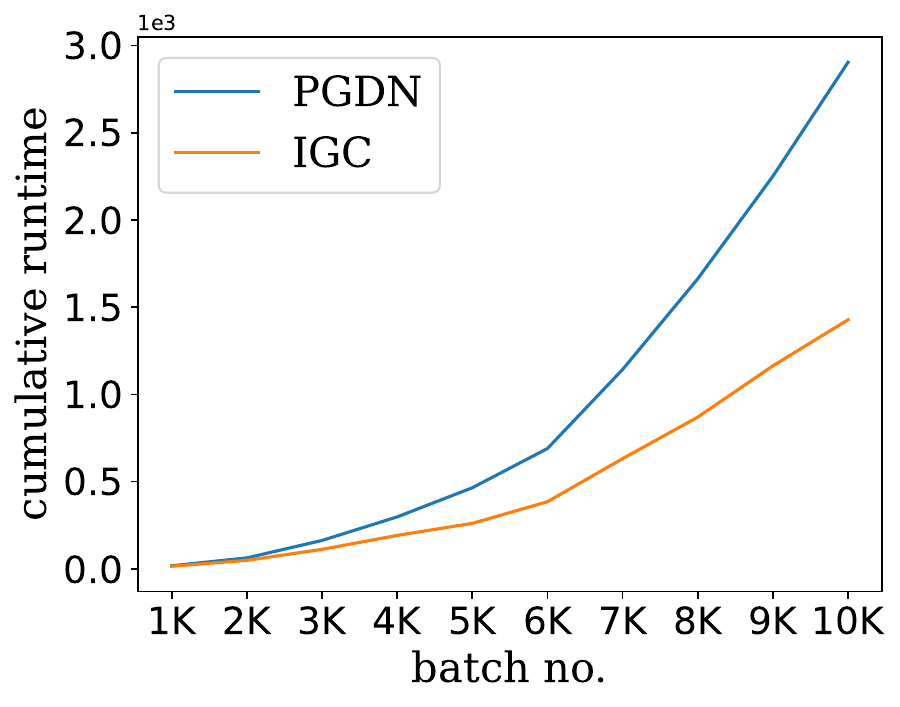}\label{fig:pokec-100-fd}}~{}~{}
  \caption{Cumulative time to add fully-dynamic batches (of size $100$). Time at $y$-axis shows the overall time taken to add all the batches so far (represented by the batch number on the $x$-axis)}\label{fig:fd-batch-100}
 \end{minipage}
%
\end{figure}

\remove{
\begin{figure*}[t!]
	\centering
	\vspace{-3ex}
	\resizebox{1\linewidth}{!}{
	\subfloat[][{\lj}]{\includegraphics[width=0.75\textwidth]{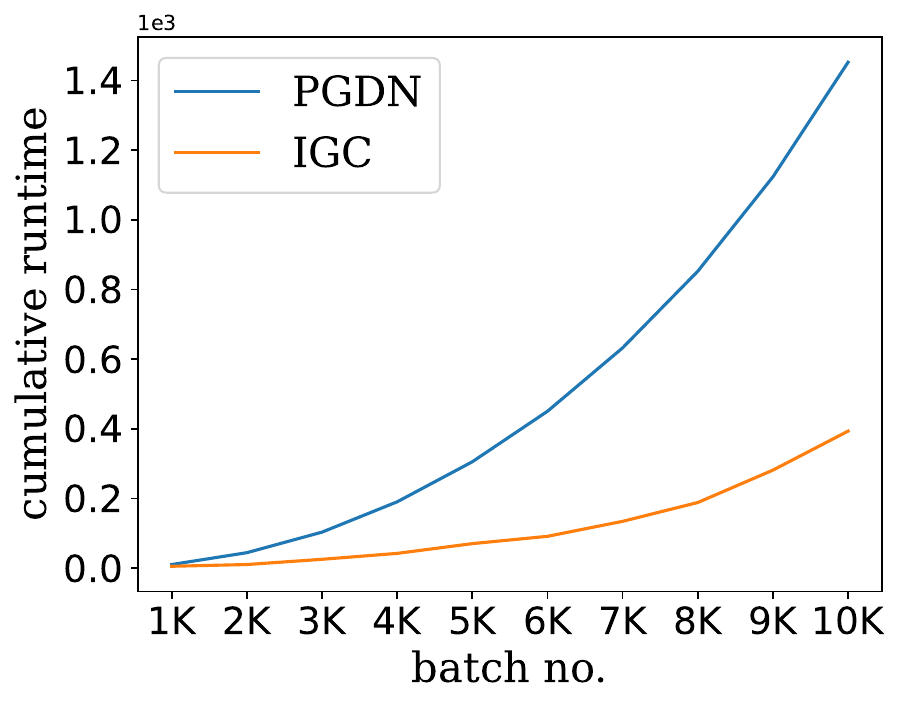}\label{fig:live-journal-fd}}~{}~{}
	\subfloat[][{\pokec}]{\includegraphics[width=0.75\textwidth]{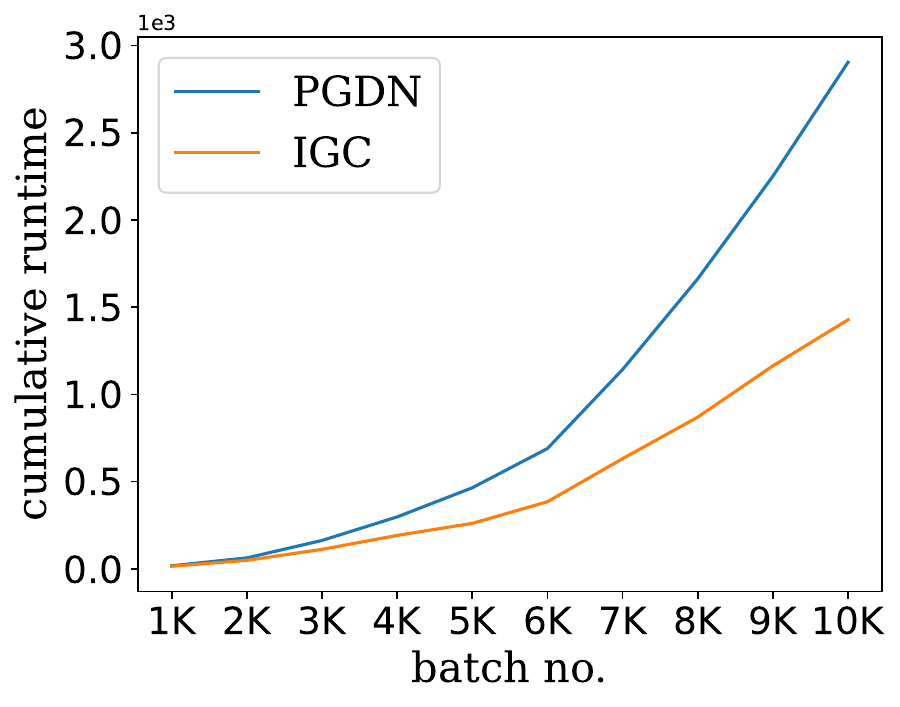}\label{fig:soc-pokec-fd}}~{}~{}
	}
	\vspace{-1ex}
	\caption{ \bf {Computation time for incremental batch of size $10$}
	}\label{fig:incr-batch-10}
	\vspace{-3ex}
\end{figure*}
}
\section{Conclusion}

In this work, we have developed an efficient algorithm for updating the counts of size-$4$ induced graphlets. Our algorithm is based on the exploration of subgraphs local to change in the graph by the fully dynamic stream of edges. We experimentally evaluated our algorithm to show the efficiency over the baseline. 

\bibliographystyle{ACM-Reference-Format}
\bibliography{references}


\begin{thebibliography}{31}


\ifx \showCODEN    \undefined \def \showCODEN     #1{\unskip}     \fi
\ifx \showDOI      \undefined \def \showDOI       #1{#1}\fi
\ifx \showISBNx    \undefined \def \showISBNx     #1{\unskip}     \fi
\ifx \showISBNxiii \undefined \def \showISBNxiii  #1{\unskip}     \fi
\ifx \showISSN     \undefined \def \showISSN      #1{\unskip}     \fi
\ifx \showLCCN     \undefined \def \showLCCN      #1{\unskip}     \fi
\ifx \shownote     \undefined \def \shownote      #1{#1}          \fi
\ifx \showarticletitle \undefined \def \showarticletitle #1{#1}   \fi
\ifx \showURL      \undefined \def \showURL       {\relax}        \fi
\providecommand\bibfield[2]{#2}
\providecommand\bibinfo[2]{#2}
\providecommand\natexlab[1]{#1}
\providecommand\showeprint[2][]{arXiv:#2}

\bibitem[Ahmed et~al\mbox{.}(2015a)]%
        {graphlet-exact-1}
\bibfield{author}{\bibinfo{person}{Nesreen~K Ahmed}, \bibinfo{person}{Jennifer
  Neville}, \bibinfo{person}{Ryan~A Rossi}, {and} \bibinfo{person}{Nick
  Duffield}.} \bibinfo{year}{2015}\natexlab{a}.
\newblock \showarticletitle{Efficient graphlet counting for large networks}. In
  \bibinfo{booktitle}{\emph{2015 IEEE international conference on data
  mining}}. IEEE, \bibinfo{pages}{1--10}.
\newblock


\bibitem[Ahmed et~al\mbox{.}(2015b)]%
        {ANRD-ICDM-2015}
\bibfield{author}{\bibinfo{person}{Nesreen~K Ahmed}, \bibinfo{person}{Jennifer
  Neville}, \bibinfo{person}{Ryan~A Rossi}, {and} \bibinfo{person}{Nick
  Duffield}.} \bibinfo{year}{2015}\natexlab{b}.
\newblock \showarticletitle{Efficient graphlet counting for large networks}. In
  \bibinfo{booktitle}{\emph{2015 IEEE International Conference on Data
  Mining}}. IEEE, \bibinfo{pages}{1--10}.
\newblock


\bibitem[Bera and Seshadhri(2020)]%
        {triangle-counting-stream-1}
\bibfield{author}{\bibinfo{person}{Suman~K Bera} {and} \bibinfo{person}{C
  Seshadhri}.} \bibinfo{year}{2020}\natexlab{}.
\newblock \showarticletitle{How the degeneracy helps for triangle counting in
  graph streams}. In \bibinfo{booktitle}{\emph{Proceedings of the 39th ACM
  SIGMOD-SIGACT-SIGAI Symposium on Principles of Database Systems}}.
  \bibinfo{pages}{457--467}.
\newblock


\bibitem[Chen et~al\mbox{.}(2022)]%
        {dynamic-cycle-1}
\bibfield{author}{\bibinfo{person}{Justin~Y Chen}, \bibinfo{person}{Talya
  Eden}, \bibinfo{person}{Piotr Indyk}, \bibinfo{person}{Honghao Lin},
  \bibinfo{person}{Shyam Narayanan}, \bibinfo{person}{Ronitt Rubinfeld},
  \bibinfo{person}{Sandeep Silwal}, \bibinfo{person}{Tal Wagner},
  \bibinfo{person}{David~P Woodruff}, {and} \bibinfo{person}{Michael Zhang}.}
  \bibinfo{year}{2022}\natexlab{}.
\newblock \showarticletitle{Triangle and Four Cycle Counting with Predictions
  in Graph Streams}. In \bibinfo{booktitle}{\emph{International Conference on
  Learning Representations}}.
\newblock


\bibitem[Danisch et~al\mbox{.}(2018)]%
        {k-clique-listing}
\bibfield{author}{\bibinfo{person}{Maximilien Danisch}, \bibinfo{person}{Oana
  Balalau}, {and} \bibinfo{person}{Mauro Sozio}.}
  \bibinfo{year}{2018}\natexlab{}.
\newblock \showarticletitle{Listing k-cliques in sparse real-world graphs}. In
  \bibinfo{booktitle}{\emph{Proceedings of the 2018 World Wide Web
  Conference}}. \bibinfo{pages}{589--598}.
\newblock


\bibitem[Das et~al\mbox{.}(2020)]%
        {fully-dynamic-clique}
\bibfield{author}{\bibinfo{person}{Apurba Das}, \bibinfo{person}{Seyed-Vahid
  Sanei-Mehri}, {and} \bibinfo{person}{Srikanta Tirthapura}.}
  \bibinfo{year}{2020}\natexlab{}.
\newblock \showarticletitle{Shared-memory parallel maximal clique enumeration
  from static and dynamic graphs}.
\newblock \bibinfo{journal}{\emph{ACM Transactions on Parallel Computing
  (TOPC)}} \bibinfo{volume}{7}, \bibinfo{number}{1} (\bibinfo{year}{2020}),
  \bibinfo{pages}{1--28}.
\newblock


\bibitem[Dhulipala et~al\mbox{.}(2021)]%
        {batch-dynamic-clique}
\bibfield{author}{\bibinfo{person}{Laxman Dhulipala},
  \bibinfo{person}{Quanquan~C Liu}, \bibinfo{person}{Julian Shun}, {and}
  \bibinfo{person}{Shangdi Yu}.} \bibinfo{year}{2021}\natexlab{}.
\newblock \showarticletitle{Parallel batch-dynamic k-clique counting}. In
  \bibinfo{booktitle}{\emph{Symposium on Algorithmic Principles of Computer
  Systems (APOCS)}}. SIAM, \bibinfo{pages}{129--143}.
\newblock


\bibitem[Hanauer et~al\mbox{.}(2022)]%
        {fully-dynamic-graphlet}
\bibfield{author}{\bibinfo{person}{Kathrin Hanauer}, \bibinfo{person}{Monika
  Henzinger}, {and} \bibinfo{person}{Qi~Cheng Hua}.}
  \bibinfo{year}{2022}\natexlab{}.
\newblock \showarticletitle{Fully Dynamic Four-Vertex Subgraph Counting}. In
  \bibinfo{booktitle}{\emph{1st Symposium on Algorithmic Foundations of Dynamic
  Networks}}.
\newblock


\bibitem[Harchaoui and Bach(2007)]%
        {HB-CVPR-2007}
\bibfield{author}{\bibinfo{person}{Zaid Harchaoui} {and}
  \bibinfo{person}{Francis Bach}.} \bibinfo{year}{2007}\natexlab{}.
\newblock \showarticletitle{Image classification with segmentation graph
  kernels}. In \bibinfo{booktitle}{\emph{2007 IEEE Conference on Computer
  Vision and Pattern Recognition}}. IEEE, \bibinfo{pages}{1--8}.
\newblock


\bibitem[Janssen et~al\mbox{.}(2012)]%
        {janssen2012model}
\bibfield{author}{\bibinfo{person}{Jeannette Janssen}, \bibinfo{person}{Matt
  Hurshman}, {and} \bibinfo{person}{Nauzer Kalyaniwalla}.}
  \bibinfo{year}{2012}\natexlab{}.
\newblock \showarticletitle{Model selection for social networks using
  graphlets}.
\newblock \bibinfo{journal}{\emph{Internet Mathematics}} \bibinfo{volume}{8},
  \bibinfo{number}{4} (\bibinfo{year}{2012}), \bibinfo{pages}{338--363}.
\newblock


\bibitem[Jha et~al\mbox{.}(2015)]%
        {path-sampling}
\bibfield{author}{\bibinfo{person}{Madhav Jha}, \bibinfo{person}{C Seshadhri},
  {and} \bibinfo{person}{Ali Pinar}.} \bibinfo{year}{2015}\natexlab{}.
\newblock \showarticletitle{Path sampling: A fast and provable method for
  estimating 4-vertex subgraph counts}. In
  \bibinfo{booktitle}{\emph{Proceedings of the 24th international conference on
  world wide web}}. \bibinfo{pages}{495--505}.
\newblock


\bibitem[Leskovec and Krevl(2014)]%
        {snapnets}
\bibfield{author}{\bibinfo{person}{Jure Leskovec} {and} \bibinfo{person}{Andrej
  Krevl}.} \bibinfo{year}{2014}\natexlab{}.
\newblock \bibinfo{title}{{SNAP Datasets}: {Stanford} Large Network Dataset
  Collection}.
\newblock \bibinfo{howpublished}{\url{http://snap.stanford.edu/data}}.
\newblock


\bibitem[McGregor and Vorotnikova(2020)]%
        {triangle-four-cycle-counting}
\bibfield{author}{\bibinfo{person}{Andrew McGregor} {and}
  \bibinfo{person}{Sofya Vorotnikova}.} \bibinfo{year}{2020}\natexlab{}.
\newblock \showarticletitle{Triangle and four cycle counting in the data stream
  model}. In \bibinfo{booktitle}{\emph{Proceedings of the 39th ACM
  SIGMOD-SIGACT-SIGAI Symposium on Principles of Database Systems}}.
  \bibinfo{pages}{445--456}.
\newblock


\bibitem[Milenkovi{\'c} et~al\mbox{.}(2010)]%
        {milenkovic2010systems}
\bibfield{author}{\bibinfo{person}{Tijana Milenkovi{\'c}},
  \bibinfo{person}{Vesna Memi{\v{s}}evi{\'c}}, \bibinfo{person}{Anand~K
  Ganesan}, {and} \bibinfo{person}{Nata{\v{s}}a Pr{\v{z}}ulj}.}
  \bibinfo{year}{2010}\natexlab{}.
\newblock \showarticletitle{Systems-level cancer gene identification from
  protein interaction network topology applied to melanogenesis-related
  functional genomics data}.
\newblock \bibinfo{journal}{\emph{Journal of the Royal Society Interface}}
  \bibinfo{volume}{7}, \bibinfo{number}{44} (\bibinfo{year}{2010}),
  \bibinfo{pages}{423--437}.
\newblock


\bibitem[Milenkovi{\'c} and Pr{\v{z}}ulj(2008)]%
        {milenkovic2008uncovering}
\bibfield{author}{\bibinfo{person}{Tijana Milenkovi{\'c}} {and}
  \bibinfo{person}{Nata{\v{s}}a Pr{\v{z}}ulj}.}
  \bibinfo{year}{2008}\natexlab{}.
\newblock \showarticletitle{Uncovering biological network function via graphlet
  degree signatures}.
\newblock \bibinfo{journal}{\emph{Cancer informatics}}  \bibinfo{volume}{6}
  (\bibinfo{year}{2008}), \bibinfo{pages}{CIN--S680}.
\newblock


\bibitem[Mukherjee et~al\mbox{.}(2018)]%
        {mukherjee2018counting}
\bibfield{author}{\bibinfo{person}{Kingshuk Mukherjee},
  \bibinfo{person}{Md~Mahmudul Hasan}, \bibinfo{person}{Christina Boucher},
  {and} \bibinfo{person}{Tamer Kahveci}.} \bibinfo{year}{2018}\natexlab{}.
\newblock \showarticletitle{Counting motifs in dynamic networks}.
\newblock \bibinfo{journal}{\emph{BMC systems biology}} \bibinfo{volume}{12},
  \bibinfo{number}{1} (\bibinfo{year}{2018}), \bibinfo{pages}{1--12}.
\newblock


\bibitem[Pinar et~al\mbox{.}(2017)]%
        {PSV-WWW-2017}
\bibfield{author}{\bibinfo{person}{Ali Pinar}, \bibinfo{person}{C Seshadhri},
  {and} \bibinfo{person}{Vaidyanathan Vishal}.}
  \bibinfo{year}{2017}\natexlab{}.
\newblock \showarticletitle{Escape: Efficiently counting all 5-vertex
  subgraphs}. In \bibinfo{booktitle}{\emph{Proceedings of the 26th
  international conference on world wide web}}. \bibinfo{pages}{1431--1440}.
\newblock


\bibitem[Pr{\v{z}}ulj et~al\mbox{.}(2004)]%
        {prvzulj2004modeling}
\bibfield{author}{\bibinfo{person}{Natasa Pr{\v{z}}ulj},
  \bibinfo{person}{Derek~G Corneil}, {and} \bibinfo{person}{Igor Jurisica}.}
  \bibinfo{year}{2004}\natexlab{}.
\newblock \showarticletitle{Modeling interactome: scale-free or geometric?}
\newblock \bibinfo{journal}{\emph{Bioinformatics}} \bibinfo{volume}{20},
  \bibinfo{number}{18} (\bibinfo{year}{2004}), \bibinfo{pages}{3508--3515}.
\newblock


\bibitem[Rahman et~al\mbox{.}(2014)]%
        {RBA-TKDE-2014}
\bibfield{author}{\bibinfo{person}{Mahmudur Rahman},
  \bibinfo{person}{Mansurul~Alam Bhuiyan}, {and} \bibinfo{person}{Mohammad
  Al~Hasan}.} \bibinfo{year}{2014}\natexlab{}.
\newblock \showarticletitle{Graft: An efficient graphlet counting method for
  large graph analysis}.
\newblock \bibinfo{journal}{\emph{IEEE Transactions on Knowledge and Data
  Engineering}} \bibinfo{volume}{26}, \bibinfo{number}{10}
  (\bibinfo{year}{2014}), \bibinfo{pages}{2466--2478}.
\newblock


\bibitem[Rossi et~al\mbox{.}(2018)]%
        {graphlet-estimate-1}
\bibfield{author}{\bibinfo{person}{Ryan~A Rossi}, \bibinfo{person}{Rong Zhou},
  {and} \bibinfo{person}{Nesreen~K Ahmed}.} \bibinfo{year}{2018}\natexlab{}.
\newblock \showarticletitle{Estimation of graphlet counts in massive networks}.
\newblock \bibinfo{journal}{\emph{IEEE transactions on neural networks and
  learning systems}} \bibinfo{volume}{30}, \bibinfo{number}{1}
  (\bibinfo{year}{2018}), \bibinfo{pages}{44--57}.
\newblock


\bibitem[Saha and Hasan(2015)]%
        {SH-CN-2015}
\bibfield{author}{\bibinfo{person}{Tanay~Kumar Saha} {and}
  \bibinfo{person}{Mohammad~Al Hasan}.} \bibinfo{year}{2015}\natexlab{}.
\newblock \showarticletitle{Finding network motifs using MCMC sampling}.
\newblock  (\bibinfo{year}{2015}), \bibinfo{pages}{13--24}.
\newblock


\bibitem[Seshadhri et~al\mbox{.}(2013)]%
        {SPT-ICDM-2013}
\bibfield{author}{\bibinfo{person}{Comandur Seshadhri}, \bibinfo{person}{Ali
  Pinar}, {and} \bibinfo{person}{Tamara~G Kolda}.}
  \bibinfo{year}{2013}\natexlab{}.
\newblock \showarticletitle{Triadic measures on graphs: The power of wedge
  sampling}. In \bibinfo{booktitle}{\emph{Proceedings of the 2013 SIAM
  international conference on data mining}}. SIAM, \bibinfo{pages}{10--18}.
\newblock


\bibitem[Shervashidze et~al\mbox{.}(2009)]%
        {shervashidze2009efficient}
\bibfield{author}{\bibinfo{person}{Nino Shervashidze}, \bibinfo{person}{SVN
  Vishwanathan}, \bibinfo{person}{Tobias Petri}, \bibinfo{person}{Kurt
  Mehlhorn}, {and} \bibinfo{person}{Karsten Borgwardt}.}
  \bibinfo{year}{2009}\natexlab{}.
\newblock \showarticletitle{Efficient graphlet kernels for large graph
  comparison}. In \bibinfo{booktitle}{\emph{Artificial intelligence and
  statistics}}. PMLR, \bibinfo{pages}{488--495}.
\newblock


\bibitem[Shin et~al\mbox{.}(2020)]%
        {fully-dynamic-triangle-count}
\bibfield{author}{\bibinfo{person}{Kijung Shin}, \bibinfo{person}{Sejoon Oh},
  \bibinfo{person}{Jisu Kim}, \bibinfo{person}{Bryan Hooi}, {and}
  \bibinfo{person}{Christos Faloutsos}.} \bibinfo{year}{2020}\natexlab{}.
\newblock \showarticletitle{Fast, accurate and provable triangle counting in
  fully dynamic graph streams}.
\newblock \bibinfo{journal}{\emph{ACM Transactions on Knowledge Discovery from
  Data (TKDD)}} \bibinfo{volume}{14}, \bibinfo{number}{2}
  (\bibinfo{year}{2020}), \bibinfo{pages}{1--39}.
\newblock


\bibitem[Stefani et~al\mbox{.}(2017)]%
        {dynamic-cycle-2}
\bibfield{author}{\bibinfo{person}{Lorenzo~De Stefani},
  \bibinfo{person}{Alessandro Epasto}, \bibinfo{person}{Matteo Riondato}, {and}
  \bibinfo{person}{Eli Upfal}.} \bibinfo{year}{2017}\natexlab{}.
\newblock \showarticletitle{Triest: Counting local and global triangles in
  fully dynamic streams with fixed memory size}.
\newblock \bibinfo{journal}{\emph{ACM Transactions on Knowledge Discovery from
  Data (TKDD)}} \bibinfo{volume}{11}, \bibinfo{number}{4}
  (\bibinfo{year}{2017}), \bibinfo{pages}{1--50}.
\newblock


\bibitem[Vacic et~al\mbox{.}(2010)]%
        {VI+-JCB-2010}
\bibfield{author}{\bibinfo{person}{Vladimir Vacic}, \bibinfo{person}{Lilia~M
  Iakoucheva}, \bibinfo{person}{Stefano Lonardi}, {and}
  \bibinfo{person}{Predrag Radivojac}.} \bibinfo{year}{2010}\natexlab{}.
\newblock \showarticletitle{Graphlet kernels for prediction of functional
  residues in protein structures}.
\newblock \bibinfo{journal}{\emph{Journal of Computational Biology}}
  \bibinfo{volume}{17}, \bibinfo{number}{1} (\bibinfo{year}{2010}),
  \bibinfo{pages}{55--72}.
\newblock


\bibitem[Wang et~al\mbox{.}(2012)]%
        {WH+-TCBB-2012}
\bibfield{author}{\bibinfo{person}{Jianxin Wang}, \bibinfo{person}{Yuannan
  Huang}, \bibinfo{person}{Fang-Xiang Wu}, {and} \bibinfo{person}{Yi Pan}.}
  \bibinfo{year}{2012}\natexlab{}.
\newblock \showarticletitle{Symmetry compression method for discovering network
  motifs}.
\newblock \bibinfo{journal}{\emph{IEEE/ACM Transactions on Computational
  Biology and Bioinformatics}} \bibinfo{volume}{9}, \bibinfo{number}{6}
  (\bibinfo{year}{2012}), \bibinfo{pages}{1776--1789}.
\newblock


\bibitem[Wang et~al\mbox{.}(2014)]%
        {WL+-TKDD-2014}
\bibfield{author}{\bibinfo{person}{Pinghui Wang}, \bibinfo{person}{John~CS
  Lui}, \bibinfo{person}{Bruno Ribeiro}, \bibinfo{person}{Don Towsley},
  \bibinfo{person}{Junzhou Zhao}, {and} \bibinfo{person}{Xiaohong Guan}.}
  \bibinfo{year}{2014}\natexlab{}.
\newblock \showarticletitle{Efficiently estimating motif statistics of large
  networks}.
\newblock \bibinfo{journal}{\emph{ACM Transactions on Knowledge Discovery from
  Data (TKDD)}} \bibinfo{volume}{9}, \bibinfo{number}{2}
  (\bibinfo{year}{2014}), \bibinfo{pages}{1--27}.
\newblock


\bibitem[Wang et~al\mbox{.}(2017)]%
        {WZ+-TKDE-2017}
\bibfield{author}{\bibinfo{person}{Pinghui Wang}, \bibinfo{person}{Junzhou
  Zhao}, \bibinfo{person}{Xiangliang Zhang}, \bibinfo{person}{Zhenguo Li},
  \bibinfo{person}{Jiefeng Cheng}, \bibinfo{person}{John~CS Lui},
  \bibinfo{person}{Don Towsley}, \bibinfo{person}{Jing Tao}, {and}
  \bibinfo{person}{Xiaohong Guan}.} \bibinfo{year}{2017}\natexlab{}.
\newblock \showarticletitle{MOSS-5: A fast method of approximating counts of
  5-node graphlets in large graphs}.
\newblock \bibinfo{journal}{\emph{IEEE Transactions on Knowledge and Data
  Engineering}} \bibinfo{volume}{30}, \bibinfo{number}{1}
  (\bibinfo{year}{2017}), \bibinfo{pages}{73--86}.
\newblock


\bibitem[Zhang et~al\mbox{.}(2013)]%
        {ZS+-CVPR-2013}
\bibfield{author}{\bibinfo{person}{Luming Zhang}, \bibinfo{person}{Mingli
  Song}, \bibinfo{person}{Zicheng Liu}, \bibinfo{person}{Xiao Liu},
  \bibinfo{person}{Jiajun Bu}, {and} \bibinfo{person}{Chun Chen}.}
  \bibinfo{year}{2013}\natexlab{}.
\newblock \showarticletitle{Probabilistic graphlet cut: Exploiting spatial
  structure cue for weakly supervised image segmentation}. In
  \bibinfo{booktitle}{\emph{Proceedings of the IEEE conference on computer
  vision and pattern recognition}}. \bibinfo{pages}{1908--1915}.
\newblock


\bibitem[Zhang et~al\mbox{.}(2012)]%
        {ZS+-TIP-2012}
\bibfield{author}{\bibinfo{person}{Luming Zhang}, \bibinfo{person}{Mingli
  Song}, \bibinfo{person}{Qi Zhao}, \bibinfo{person}{Xiao Liu},
  \bibinfo{person}{Jiajun Bu}, {and} \bibinfo{person}{Chun Chen}.}
  \bibinfo{year}{2012}\natexlab{}.
\newblock \showarticletitle{Probabilistic graphlet transfer for photo
  cropping}.
\newblock \bibinfo{journal}{\emph{IEEE Transactions on Image Processing}}
  \bibinfo{volume}{22}, \bibinfo{number}{2} (\bibinfo{year}{2012}),
  \bibinfo{pages}{802--815}.
\newblock


\end{thebibliography}

\end{document}